\documentclass{article}

\usepackage[preprint]{neurips_2026}

\usepackage[utf8]{inputenc}
\usepackage[T1]{fontenc}
\usepackage{microtype}
\usepackage{graphicx}
\usepackage{subfigure}
\usepackage{booktabs}
\usepackage{hyperref}
\usepackage{url}
\usepackage{nicefrac}
\usepackage{xcolor}
\usepackage{natbib}

\usepackage{amsmath}
\usepackage{amssymb}
\usepackage{mathtools}
\usepackage{amsthm}

\usepackage{multirow}

\usepackage{algorithm}
\usepackage{algorithmic}

\usepackage[capitalize,noabbrev]{cleveref}

\usepackage{tikz}
\usetikzlibrary{patterns}

\theoremstyle{plain}
\newtheorem{theorem}{Theorem}[section]

\newtheorem{lemma}[theorem]{Lemma}

\theoremstyle{definition}
\newtheorem{definition}[theorem]{Definition}
\newtheorem{assumption}[theorem]{Assumption}
\theoremstyle{remark}

\newcommand{\Cmin}{C_{\min}}
\newcommand{\Hsol}{H_{\text{sol}}}
\newcommand{\Ecut}{E_{\text{cut}}}
\newcommand{\Sol}{\text{Sol}}
\newcommand{\taustar}{\tau^*}

\title{How Task Structure Limits Multi-Agent Success:\\ An Information-Theoretic Analysis}

\author{%
  Shi Pan\thanks{Corresponding author.} \\
  University College London \\
  London, UK \\
  \texttt{shi.pan@ucl.ac.uk} \\
  \And
  Ming Luo \\
  University of Bristol \\
  Bristol, UK \\
  \texttt{ming.luo@bristol.ac.uk} \\
}

\begin{document}

\maketitle

\begin{abstract}
Multi-agent systems (MAS) were expected to overcome the limitation of single-agent systems (SAS) through collaboration. However, under typicality conditions on the task's constraint graph and bounded inter-agent communication, we prove that the success probability of a MAS is closely tied to the connectivity of task constraints, where each agent has limited information-processing capacity. Specifically, the success probability decays exponentially with an information bottleneck that emerges from partitioning the task's constraint graph among agents. We define this quantity as the \emph{minimum cut cost} $C_{\min}$ of the potential constraint graph of each task. This information-theoretic bound applies to both open systems with external feedback and closed systems without. We validate our theory on both synthetic experiments and real-world empirical data from SWE-bench submissions. From our framework, effective MAS design should incorporate task-inherent constraints alongside engineering optimization, and when $\Cmin$ is high, practitioners should restructure tasks rather than simply scaling agents or communication.
\end{abstract}

\section{Introduction}
\label{sec:intro}

Multi-agent systems (MAS) have garnered significant attention from both academia and industry~\citep{guo2024large,tran2025multiagent,zhou2024survey}. 
The premise is intuitive: multiple agents, each follows task oriented contexts and contributing distinct capabilities, should collectively solve problems beyond the reach of any individual~\citep{du2024improving,wu2024autogen}. 
This promise of emergent collective intelligence has fueled rapid proliferation of MAS architectures~\cite{hong2024metagpt,qian2024chatdev,wu2024autogen,chen2024agentverse}. 
Yet a puzzling pattern has emerged---MAS frequently underperform on complex tasks, and in some cases fail to match even single-agent baselines~\citep{cemri2025why,wang2024rethinking,kim2025towards,gao2025singleagent}. 
What explains this gap between expectation and reality? 
More fundamentally, are there theoretical limits that constrain what collaborative agents can achieve~\citep{kim2025towards}?

How to effectively collaborate between multiple agents is highly related with its communication strategy.
For instance, task-specific prompt show some successes on MAS which optimize prompts for individual subtasks and substantially improve agent performance at each execution step \citep{wang2024promptagent,hong2024metagpt,qian2024chatdev}. 
However, limited context windows and the risk of quality degradation with long-context force MAS can not share full conext history for all sub-agents. \citep{zhang2024chain,guo2024large,liu2024lost,du2025context}. 
In response, many MAS architectures adopt a simplifying assumption: sub-agents operate independently without direct inter-agent communication \citep{tran2025multiagent}. 
But coordination cannot be entirely eliminated—some agent must allocate tasks, resolve conflicts, and integrate partial results through local communication \citep{yokoo1998distributed}. 
Existing approaches typically circumvent this challenge by introducing a stronger model as a centralized planner to orchestrate agent behavior \citep{wu2024autogen,zhang2024coa,qian2025scaling}. 
While pragmatic, this strategy represents a compromise rather than a solution—it effectively offloads the coordination problem to a more capable single agent, sidestepping rather than resolving the fundamental question of how peer agents can collaborate effectively \citep{xi2025scaling,tran2025multiagent}. 
More critically, this coordination overhead is not merely incidental \citep{zhang2025agentprune,chen2025c2c}. 

In real-world tasks, constraints that span multiple components—requiring information to flow across boundaries—become the critical bottlenecks~\cite{chen2024scalable,tran2025multiagent,ding2024seqcomm}. The difficulty of a task is defined by these bottlenecks as they determine the irreducible coordination cost that no system design can circumvent~\cite{cemri2025why,shen2024taskbench,niu2025flow}. Centralized stronger planner does not eliminate these dependencies; it merely relocates them~\cite{chen2024scalable,yang2025agentnet}. By following this idea, recent works suggests fundamental limits of multi-agent success are thus set not by how agents are organized~\cite{hong2024metagpt,cemri2025why,tran2025multiagent}. In this paper, we suggests:

\emph{Task structure, rather than system configuration, determines some fundamental limits of MAS success}.

We formalize this insight through the \emph{ constraint graph of a task} (CGT), where edges capture how satisfying one constraint affects the feasibility of others. Partitioning CGT across agents inevitably incurs information loss. We then define the cost of partition process by using minimum information losses $\Cmin$. Within an information-theoretic framework, our central result establishes that, under typicality conditions on the constraint graph, success probability decays exponentially with this quantity: $P(\text{success}) = O(2^{-\Omega(\Cmin)})$ (Theorem~\ref{thm:scale}). We derive complementary bounds based on graph diameter (Theorem~\ref{thm:diameter}) and information budget trade-offs (Theorem~\ref{thm:budget}), together characterizing how task structure constrains the success rate of MAS. Beyond the qualitative ``coupling penalty'' intuition, our contributions are quantitative: \emph{(i)} a closed-form exponential decay rate $\Omega(\Cmin)$ tied to task structure rather than algorithmic choice; \emph{(ii)} an explicit phase-transition threshold $F + R = c \cdot \Cmin - \Hsol$ separating feasible from exponentially hard regimes (Theorem~\ref{thm:budget}); and \emph{(iii)} a precise gap between open and closed systems, showing that internal communication $M$ alone cannot bridge the information gap (Lemma~\ref{lem:conservation}).

We validate this framework through a two-level strategy. Synthetic experiments on controlled agents confirm the mathematical structure of our bounds. Empirical analysis of SWE-bench submissions involves 46 real agent systems across 300 software engineering tasks to demonstrates the qualitative relationships.

\section{Related Works}
\label{sec:relatedworks}

LLM-based MAS have become a promising paradigm for complex problem solving \citep{guo2024large,tran2025multiagent,zhou2024survey}, with frameworks such as MetaGPT \citep{hong2024metagpt}, ChatDev \citep{qian2024chatdev}, AutoGen \citep{wu2024autogen}, and AgentVerse \citep{chen2024agentverse} deployed across software engineering \citep{jimenez2024swebench}, scientific discovery \citep{lu2024aiscientist}, and embodied tasks \citep{zhang2024building}. Despite the premise of emergent collective intelligence \citep{du2024improving}, empirical evidence reveals that MAS frequently underperform single-agent baselines \citep{cemri2025why,wang2024rethinking,kim2025towards,gao2025singleagent}. Existing explanations cite coordination overhead \citep{zhang2025agentprune,chen2025c2c}, error propagation \citep{tran2025multiagent,xi2025scaling}, and role redundancy \citep{wang2024rethinking}, but identify symptoms without a unified theory of \emph{when} and \emph{why} MAS must fail.

Our framework draws on three complementary traditions. Individual agents face bounded information-processing capacity even in nominally large context windows \citep{liu2024lost,du2025context,schnabel2025lost,guo2025unable,simon1990bounded}; we abstract this as the effective capacity $W$. Distributed constraint satisfaction problems \citep{yokoo1998distributed, modi2005adopt, faltings2006distributed} and communication complexity \citep{yao1979some, kushilevitz1997communication, mezard2009information, molloy2014frozen} establish that constraint structures across agents impose irreducible coordination cost---a cost not eliminated by introducing centralized planners \citep{wu2024autogen,zhang2024coa,qian2025scaling,xi2025scaling,tran2025multiagent}. We use network information theory tools---cutset bounds \citep{thakor2016cutset, el2011network} and the data processing inequality \citep{cover2006elements}---to quantify information loss across the partition, complementing distributed information bottleneck approaches \citep{tishby1999information, wang2025tadib, ding2024magi}.

\section{Problem Formulation}
\label{sec:formulation}


\subsection{Task Representation: Constraint Graphs in Bottleneck-Centric View}

Complex real-world tasks typically decompose into manageable subproblems 
connected by coordination requirements. Our framework focuses on the 
\emph{irreducible bottleneck} $\tau_{\text{hard}}$: the core subproblem 
where multi-agent coordination is unavoidable and no exploitable special 
structure exists. This bottleneck-centric view parallels established 
analysis paradigms—critical path methods in project scheduling, min-cut 
characterizations in network flow, and frozen-core analysis in constraint 
satisfaction~\citep{kelley1959critical, mezard2009information}.

By following this direction, we represent the intrinsic constraints of a task as a constraint graph $G_C = (V, E)$, where $V$ is the set of variables to be determined, $E$ represents dependency relationships between variables, each edge $e \in E$ carries a potential function $\phi_e$ encoding the constraint with entropy $H(\phi_e)$, and $\Sol$ denotes the set of assignments satisfying all constraints. 

This graph represents the \emph{irreducible bottleneck} by following typicality conditions that capture local dependence, bounded correlation, probability concentration, and entropy concentration. While many problems directly satisfy these conditions (e.g., random k-SAT, random graph coloring)\citep{chen2004normal, xu2000exact}, which serve as sufficient conditions for tight bounds, our bound is designed to identify the true bottlenecks in complex real-world problems by filtering out inherent redundancies in practical settings. And the graph structure naturally captures constraint satisfaction problems (CSP), combinatorial optimization, and structured reasoning tasks.

\emph{Information non-redundancy (T1, T1').} Condition T1 requires that cut-edge constraints carry non-redundant information: $I(\taustar; \Ecut \mid \text{Internal}) \geq c_{T1} \cdot H(\Ecut)$. Condition T1' bounds correlations among cut edges: $H(\Ecut) \geq c_0 \cdot \sum_{e} H(\phi_e)$, preventing hidden algebraic dependencies from collapsing joint entropy.

\emph{Solution space regularity (T2, T3').} Condition T2 excludes anomalous concentration—configurations where solutions can be guessed with probability far exceeding the entropy-based prediction $2^{-H}$. This rules out ``backdoors'' that bypass information requirements. Condition T3' requires entropy concentration around its expectation, ensuring typical-case analysis reflects actual system behavior.

Typicality conditions operationalize "no exploitable structure": 
T1/T1' ensure partitioning creates genuine information loss (not 
masked by redundancy), while T2/T3' ensure the solution space lacks 
shortcuts. A subproblem failing these conditions can be simplified 
through restructuring—it is not the true bottleneck. Our bounds thus 
characterize the fundamental limit imposed by irreducible coordination 
requirements, with overall task success probability dominated by the 
bottleneck: $P(\text{success}) \leq P(\text{bottleneck success})$.  These conditions (T1, T1', T2, T3') are detailed in Appendix~\ref{app:typicality}.

\subsection{Agent Model and Task Partitioning}

Consider $m$ agents, each assigned a distinct subgraph $G_i$ (one-to-one mapping between agents and partitions) of the constraint graph. A partitioning scheme $\mathcal{P}$ divides $G_C$ into $m$ subgraphs, necessarily creating a cut edge set $\Ecut(\mathcal{P})$ consisting of edges whose endpoints belong to different agents. We define the \emph{minimum cut cost} as:
\begin{equation}
\Cmin := \min_{\mathcal{P}} \sum_{e \in \Ecut(\mathcal{P})} H(\phi_e)
\label{eq:cmin}
\end{equation}
where the minimum is over all valid $m$-partitions. Intuitively, partitioning a task severs dependency relationships, creating information loss proportional to the entropy of cut constraints. Our framework analyzes tasks at the level of their \emph{irreducible bottleneck} $\tau_{\text{hard}}$, the core subproblem that cannot be further decomposed without information loss. The bottleneck size is denoted $n_h$.

The solution entropy $H_{\text{sol}}$ quantifies solution space size; the special case $H_{\text{sol}} = 0$ (unique solution) represents maximum constraint tightness where constraints fully determine the correct answer. Our bounds incorporate solution structure through $H_{\text{sol}}$, capturing how task ambiguity interacts with partition costs to determine fundamental limits of multi-agent success.

To accurately characterize MAS with limited resources and processing capabilities, we make three assumptions that characterize the systems our bounds address.
\begin{assumption}[Limited Agent Capacity]
\label{ass:capacity}
Each agent can effectively process at most $W$ bits of task-relevant information in a single execution: $I(\taustar; \text{Info}_i^{(t)}) \leq W$. This assumption is empirically motivated by LLM context limitations.
\end{assumption}
\begin{assumption}[Limited Communication Budget]
\label{ass:communication}
Agents have finite budget for communicating between each other and obtaining feedback from environments. The system's information budget consists of: internal communication $M$ (bits exchanged between agents), external feedback $F$ (information about $\taustar$ obtained through execution), and feedback propagation $R$ (information transmitted from execution to planning layers). The minimum cut cost $\Cmin$ as defined in Equation~\ref{eq:cmin} characterizes the information bottleneck created by task partitioning.
\end{assumption}
\begin{assumption}[Low-Entropy Solution Set]
\label{ass:solution}
The solution entropy $\Hsol := \log |\Sol|$ satisfies $\Hsol = o(n_h)$, permitting polynomial or subexponential solution counts. This covers most practical tasks including multi-solution CSPs with sub-exponential solution counts.
\end{assumption}

Notation used throughout is summarised in Appendix~\ref{app:notation}.

\section{Main Theorems}
\label{sec:theorems}

We present three complementary bounds capturing different aspects of MAS limitations.
The bounds share a common foundation: \emph{task partitioning creates an irreducible
information gap} that limits success probability. We first establish
this information gap (Lemma~\ref{lem:gap_conservation}), then derive bounds from
three perspectives: the minimum cut cost (Theorem~\ref{thm:scale}), the graph
diameter (Theorem~\ref{thm:diameter}), and the information budget trade-off between
internal communication and external feedback (Theorem~\ref{thm:budget}).

To formalize the analysis, we define the total information available to the MAS:
\begin{definition}[System Information]
\label{def:allinfo}
Let $\text{Internal} := (G_1, \ldots, G_m)$ denote the initial information from
task partitioning, where each $G_i$ is agent $i$'s assigned subgraph. The total
system information is $\text{AllInfo} := (\text{Internal}, \mathcal{M}, \mathcal{F},
\mathcal{R})$, where $\mathcal{M}$ denotes inter-agent messages, $\mathcal{F}$
denotes external feedback, and $\mathcal{R}$ denotes feedback propagation.
\end{definition}

\subsection{The Information Gap: Foundation of All Bounds}

The central insight of our framework is that partitioning the constraint graph
among agents creates an information gap. When edges cross partition
boundaries, the constraints they encode become inaccessible to any single agent,
yet satisfying these constraints requires coordinated knowledge.

Consider the cut edge set $\Ecut$: these edges connect variables assigned to
different agents. The entropy $H(\Ecut)$ measures the total information content
of cross-partition constraints. Critically, no agent fully observes $\Ecut$,
and communication cannot recover this information in closed systems.

\begin{lemma}[Information Gap Conservation]
\label{lem:gap_conservation}
Under Assumptions~\ref{ass:capacity}--\ref{ass:solution} and typicality conditions T1, T1':
\begin{equation}
H(\taustar | \text{Internal}) \geq c_1 \cdot \Cmin
\end{equation}
where $c_1 = c_{T1} \cdot c_0$ depends on typicality constants.
\end{lemma}

The information gap analysis follows the data processing inequality framework \citep{cover2006elements} and cutset bound methodology \citep{thakor2016cutset, el2011network}. Full proof in Appendix~\ref{app:gap_conservation_proof}.

This lemma establishes that \emph{the information gap is determined by task structure,
not system design}. The gap $c_1 \cdot \Cmin$ represents information that no agent
possesses initially, regardless of how sophisticated the MAS architecture may be.

\subsection{Theorem 1: Scale Bound}

The Gap Conservation Lemma establishes a lower bound on conditional entropy.
To convert this into a success probability bound, we use the small-ball
probability framework: success requires guessing $\taustar$ correctly, and
the probability of correct guessing is bounded by the posterior distribution.

\begin{theorem}[Scale Bound]
\label{thm:scale}
Under Assumptions~\ref{ass:capacity}--\ref{ass:solution} and typicality conditions, for an MAS with $m$ agents:
\begin{equation}
P(\text{success}) = O\left(2^{-\Omega(\Cmin - \Hsol)}\right)
\end{equation}
where $\Cmin = \Cmin(\tau_{\text{hard}}, m)$ is the minimum cut cost of the bottleneck task under $m$-agent partitioning.
\end{theorem}

\paragraph{Proof Architecture.}
The proof proceeds in three steps:

\textbf{Step 1 (Max posterior bound):} Success probability is bounded by the
expected maximum posterior: $P(\text{success}) \leq \mathbb{E}[\max_x P(\taustar = x
| \text{AllInfo})]$. This follows because any correct answer must have non-zero
posterior probability.

\textbf{Step 2 (Entropy-to-probability conversion):} Under condition T2 (no
anomalous concentration), the maximum posterior is polynomially related to
entropy: $\mathbb{E}[\max_x P(\taustar = x | \text{AllInfo})] \leq n^{d_0} \cdot
\mathbb{E}[2^{-H(\taustar | \text{AllInfo})}]$. This excludes ``backdoors'' where
solutions can be guessed far better than entropy predicts.

\textbf{Step 3 (Annealed entropy framework):} Define annealed entropy $\tilde{H}
:= -\log \mathbb{E}[2^{-H}]$. Under T3' (entropy concentration), $\tilde{H} =
\Theta(\bar{H})$ where $\bar{H}$ is the expected conditional entropy. By
Lemma~\ref{lem:gap_conservation}, $\bar{H} \geq c_1 \cdot \Cmin$, yielding:
\begin{equation}
P(\text{success}) \leq n^{d_0} \cdot 2^{-\tilde{H}} \leq n^{d_0} \cdot 2^{-c_2 \cdot c_1 \cdot \Cmin}
\end{equation}

When $\Cmin = \omega(d_0 \cdot \log n)$, the polynomial factor is absorbed into
the exponential. For the multi-solution case with $\Hsol > 0$, the effective gap
becomes $\Cmin - \Hsol$, yielding the general form. Full proof in Appendix~\ref{app:theorem1}.

For the important special case of closed systems with unique solutions ($\Hsol = 0$),
this simplifies to $P(\text{success}) = O(2^{-\Omega(\Cmin)})$.

\subsection{Theorem 2: Diameter Bound}

While Theorem 1 captures information loss through partition cuts, an alternative
perspective considers information decay along paths. When the constraint graph has
large diameter, information about distant variables must traverse many agent
boundaries, decaying exponentially at each step.

\paragraph{Information Processing Model.}
Agent information processing is modeled as a noisy channel. The Strong Data
Processing Inequality (SDPI) constant characterizes information retention:
\begin{equation}
\eta(K) = \sup_{p_U, p_{X|U}} \frac{I(U;Y)}{I(U;X)} < 1
\end{equation}
for any non-trivial channel $K: \mathcal{X} \to \mathcal{Y}$. Classical examples:
additive Gaussian noise gives $\eta = \text{SNR}/(1+\text{SNR})$; binary symmetric
channel with error $p$ gives $\eta = (1-2p)^2$; erasure channel gives $\eta = 1-\epsilon$.

\begin{lemma}[Path Information Decay]
\label{lem:path_decay}
For a path $P = (v_0, \ldots, v_\ell)$ with SDPI constant $\eta < 1$:
\begin{equation}
I(\taustar_{v_0}; \text{Info}_{v_\ell}) \leq \eta^\ell \cdot I(\taustar_{v_0}; \text{Info}_{v_0})
\end{equation}
\end{lemma}

\begin{theorem}[Diameter Bound]
\label{thm:diameter}
Under the same assumptions, let $D$ be the constraint graph diameter, $\eta < 1$ the SDPI constant, and $\bar{H}_e$ the average edge entropy. Then:
\begin{equation}
P(\text{success}) \leq O\left(2^{-\Omega(D \cdot (1-\eta) \cdot \bar{H}_e)}\right)
\end{equation}
\end{theorem}

This bound is tighter than Theorem 1 for sparse graphs with large diameter (e.g.,
path graphs where $\Cmin = O(m)$ but $D = \Theta(n)$).

\subsection{Theorem 3: Information Budget Trade-off}

A natural question arises: can agents overcome the information gap by communicating
more? Surprisingly, the answer is \emph{no} for closed systems---internal
communication cannot create new information about the solution, only redistribute
existing information.

\begin{lemma}[Communication Information Conservation]
\label{lem:conservation}
In closed systems ($F = R = 0$), for any number of communication rounds $T$:
\begin{equation}
I(\taustar; \text{Internal}, \mathcal{M}^{(1:T)}) = I(\taustar; \text{Internal})
\end{equation}
Communication redistributes but cannot create information about $\taustar$.
\end{lemma}

\paragraph{Protocol Scope.}
This lemma applies to \emph{non-adaptive, simultaneous-message} protocols where each agent's
message at round $t$ depends only on its local information and previously received messages,
not on the content of messages being sent in the same round by other agents.

This lemma reveals a \emph{fundamental asymmetry}: internal communication $M$
cannot bridge the information gap, but external feedback $F$ can. When an agent
executes a candidate solution and receives error feedback (e.g., ``test failed at
line 3''), this provides direct information about $\taustar$ that breaks the
closed-system Markov structure.

\begin{theorem}[Information Budget Trade-off]
\label{thm:budget}
Under the same assumptions with feedback bounds (F1, F2):

\textbf{Closed systems} ($F = R = 0$):
\begin{equation}
P(\text{success}) \leq O\left(2^{-c \cdot \Cmin + \Hsol}\right)
\end{equation}
Note: the bound is \emph{independent of communication volume $M$}.

\textbf{Open systems} ($F, R > 0$):
\begin{equation}
P(\text{success}) \leq O\left(2^{-\max(0, c \cdot \Cmin - F - R - \Hsol)}\right)
\end{equation}
where $c > 0$ depends on typicality constants.
\end{theorem}

\paragraph{Phase Transition.}
The open system bound reveals a sharp phase transition at threshold $F + R = c
\cdot \Cmin - \Hsol$:
\begin{itemize}
\item \textbf{Below threshold}: Success probability decays exponentially in the
residual gap $(c \cdot \Cmin - F - R - \Hsol)$
\item \textbf{At/above threshold}: The exponent becomes non-positive; success
becomes possible (though not guaranteed)
\end{itemize}

The total information budget decomposes as:
\begin{equation}
I(\taustar; \text{AllInfo}) \leq I(\taustar; \text{Internal}) + F + R
\end{equation}
where by Gap Conservation, $I(\taustar; \text{Internal}) \leq H(\taustar) - c_1
\cdot \Cmin$. External feedback directly compensates for the partition-induced
information gap. Full proof in Appendix~\ref{app:theorem3_proof}.

The practical implication is stark: for closed systems with high $\Cmin$, adding
communication bandwidth provides \emph{zero asymptotic benefit}. Resources should
instead focus on task restructuring or incorporating external feedback mechanisms.

\paragraph{Multi-Solution Case.}
Under Assumption~\ref{ass:solution} (low-entropy solution set, $\Hsol = o(n_h)$),
the effective information gap becomes $c \cdot \Cmin - \Hsol$. When multiple valid
solutions exist ($|\Sol| > 1$), finding \emph{any} valid solution requires less
information than finding a unique solution. The $\Hsol$ term quantifies this
reduction: with $2^{\Hsol}$ valid solutions, the search space effectively contracts.
The detailed derivation appears in Appendix~\ref{app:theorem3}, Step 3.

\begin{table}[t]
\caption{Summary of main theorems.}
\label{tab:theorems}
\vskip 0.1in
\begin{center}
\begin{small}
\begin{tabular}{lccc}
\toprule
Theorem & Focus & Parameter & Applies \\
\midrule
Scale (1) & Cut & $\Cmin$ & All \\
Diameter (2) & Path & $D$, $\eta$ & Sparse \\
Budget (3) & Feedback & $F$, $R$ & Open \\
\bottomrule
\end{tabular}
\end{small}
\end{center}
\vskip -0.1in
\end{table}

\paragraph{Theorem Selection.}
The three bounds are complementary: Theorem 1 (Scale Bound) applies universally
and is tightest for dense graphs; Theorem 2 (Diameter Bound) is tighter for sparse
graphs with large diameter; Theorem 3 (Budget Trade-off) characterizes the
phase transition for open systems. In practice, compute both $\Cmin$ and $D$ to
determine which bound provides the tighter constraint.

\section{Experimental Validation}
\label{sec:experiments}

We employ a two-level validation strategy: \textbf{synthetic experiments} test numerical consistency under controlled conditions where agents strictly satisfy our behavioral assumptions, while \textbf{empirical experiments} assess practical predictive power on real LLM agent systems.

\subsection{Experimental Setup}
\label{sec:exp_setup}

\paragraph{Synthetic Experiments.} We implement graph coloring on constraint graphs as our testbed. \textit{Agent model:} MockAgents strictly satisfy Assumptions~\ref{ass:capacity}--\ref{ass:solution}, sampling independently and uniformly for unknown constraints without implicit coordination. \textit{Graph structures:} Path, cycle, star, grid, complete, and bipartite graphs. \textit{Parameters:} $n \in \{6, 8, 10, 12, 16\}$ nodes, $k \in \{3, 4, 5\}$ colors, $m \in \{2, 3, 4\}$ agents. \textit{Trials:} 200--1000 per configuration. \textit{Metrics:} Success rate, information gap $\Delta H = \sum_i H(C_i) - H(C_{\text{global}})$, and constant $\kappa = P_{\text{exp}}/P_{\text{theory}}$.

\paragraph{Empirical Experiments (SWE-bench).} We analyze 46 agent systems (18 MAS, 18 single-agent, 10 non-agentic baselines) on SWE-bench Lite (300 software engineering tasks), totaling 14,891 result records. $\Cmin$ is estimated from ground-truth patches at line-level granularity on code dependency graphs, with statement-, expression-, and token-level as robustness checks. The full computation pipeline---graph construction, weighting scheme, partition count, statistical tests, and reproducibility configuration---is detailed as Algorithm~\ref{alg:cmin_swebench} in Appendix~\ref{app:cmin_pipeline}. \textit{Metrics:} Correlation coefficients (Pearson $r$, Spearman $\rho$), exponential fit parameters, and threshold detection. Code will be released upon acceptance.

\subsection{Results}
\label{sec:results}

\paragraph{Theorem 1 Validation (Scale Bound).} For path graphs with $k=3$ colors, the theoretical prediction $\Delta H = \Cmin \cdot \log_2(k/(k-1))$ yields constant $c = 0.585$ bits per cut edge. The measured information gap matches this analytical prediction exactly ($R^2 = 1.0$, exact constant match; Table~\ref{tab:synthetic}). Graph structure verification confirms all six $\Cmin$ formulas match exactly (Table~\ref{tab:graph_structure}). Constant stability analysis yields $\kappa = 0.98 \pm 0.03$ across 13 configurations with no systematic drift ($R^2 = 0.004$). \textit{Key insight:} Under strictly satisfied assumptions, theoretical predictions are mathematically exact.

\begin{table}[t]
\caption{Synthetic validation: Information gap scales linearly with cut cost.}
\label{tab:synthetic}
\vskip 0.1in
\begin{center}
\begin{small}
\begin{sc}
\begin{tabular}{cccc}
\toprule
$\Cmin$ & $\Delta H$ & $\Delta H / \Cmin$ & Trials \\
\midrule
1 & 0.585 & 0.585 & 40 \\
2 & 1.170 & 0.585 & 40 \\
3 & 1.755 & 0.585 & 10 \\
\bottomrule
\end{tabular}
\end{sc}
\end{small}
\end{center}
\vskip -0.1in
\end{table}

\begin{table}[t]
\caption{Graph structure $\Cmin$ verification. All formulas match exactly.}
\label{tab:graph_structure}
\vskip 0.1in
\begin{center}
\begin{small}
\begin{tabular}{lcccc}
\toprule
Graph & Formula & Theory & Actual & Match \\
\midrule
Path & $m-1$ & 1 & 1 & \checkmark \\
Cycle & $2$ (for $m=2$) & 2 & 2 & \checkmark \\
Star & $n-1-\text{leaves}$ & 4 & 4 & \checkmark \\
Grid & $\min(r,c)$ & 2 & 2 & \checkmark \\
Complete & $\lfloor n/2 \rfloor \cdot \lceil n/2 \rceil$ & 16 & 16 & \checkmark \\
Bipartite & $a \cdot b$ & 16 & 16 & \checkmark \\
\bottomrule
\end{tabular}
\end{small}
\end{center}
\vskip -0.1in
\end{table}

\paragraph{Theorem 3 Validation (Budget Trade-off).} On a $4 \times 5$ grid with $k=5$ colors ($\Cmin=4$), closed systems show success rate stable at $\approx 32\%$ regardless of communication budget $M$ ($M=0$: 31.7\%, $M=100$: 32.0\%), confirming Lemma~\ref{lem:conservation}. Open systems exhibit a sharp phase transition: success jumps from 31.8\% at $F=0$ to 93.2\% at $F=1.6$ bits, reaching 100\% for $F \geq 3.2$ bits. The critical point ($F_{\text{crit}} \approx 1.29$ bits) aligns with the theoretical threshold (Figure~\ref{fig:open_closed}). \textit{Key insight:} Communication redistributes but cannot create information; external feedback breaks this barrier.

\begin{figure}[t]
\centering
\includegraphics[width=\columnwidth]{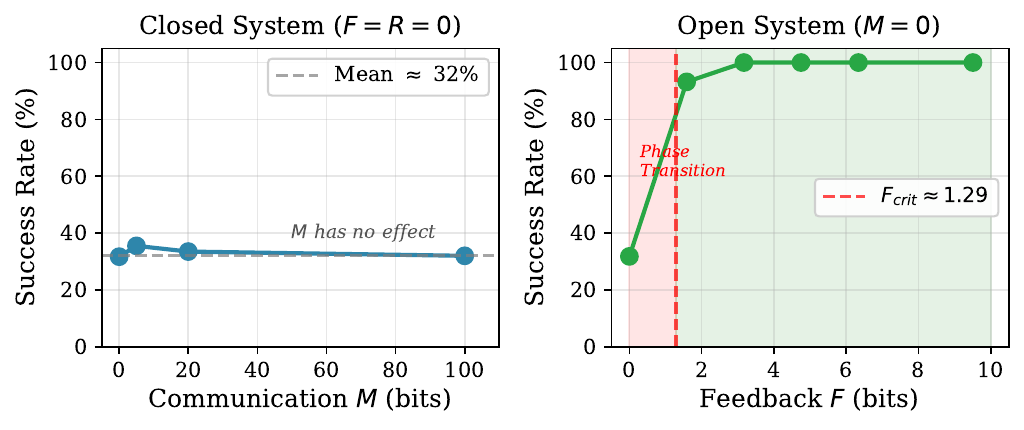}
\caption{Open vs closed system comparison on grid graph ($\Cmin=4$). \textbf{Left:} Closed system success rate vs communication $M$---remains flat at $\approx 32\%$. \textbf{Right:} Open system success rate vs feedback $F$---sharp phase transition at $F \approx 1.3$ bits.}
\label{fig:open_closed}
\end{figure}

\paragraph{Empirical Validation (SWE-bench).} We test whether the \emph{qualitative relationships}
predicted by our theorems hold for real-world software engineering tasks. Since sufficient conditions for tight bounds cannot be fully verified for code tasks, we interpret these results as testing directional predictions rather than quantitative bounds. All four qualitative predictions show statistically
significant effects in the predicted direction (Table~\ref{tab:swe}): (P1) negative correlation $r = -0.33$ ($p < 0.001$); (P2) MAS disadvantage at high $\Cmin$: $-2.3\%$ ($p = 0.017$); (P3) exponential decay $R^2 = 0.11$ ($p < 0.001$); (P4) threshold at $\Cmin \approx 6.0$ with $\rho = -0.70$ ($p = 0.025$). The moderate $R^2 = 0.11$ for exponential decay (P3) reflects several factors: our bounds are
\emph{upper bounds} (actual performance can fall below), LLM agents may exhibit implicit
coordination through shared training that violates T1, and crucially, T2 (no anomalous concentration)
is not verified for code tasks---certain code patterns may serve as ``backdoors'' that allow
solution guessing beyond entropy-based predictions. These results validate the \emph{qualitative
structure} (correlations, thresholds, comparative dynamics) while acknowledging that quantitative tightness requires verification of sufficient conditions.

\begin{table}[t]
\caption{SWE-bench: Directional predictions show statistically significant effects.}
\label{tab:swe}
\vskip 0.1in
\begin{center}
\begin{small}
\begin{tabular}{lccc}
\toprule
Prediction & Theory & Result & $p$ \\
\midrule
P1: Neg. correlation & Thm 1 & $r=-0.33$ & $<0.001$ \\
P2: No MAS gain & Thm 3 & $-2.3\%$ & $0.017$ \\
P3: Exp. decay & Thm 1 & $R^2=0.11$ & $<0.001$ \\
P4: Threshold & Budget & $\rho=-0.70$ & $0.025$ \\
\bottomrule
\end{tabular}
\end{small}
\end{center}
\vskip -0.1in
\end{table}

\paragraph{Validation Summary.} Table~\ref{tab:validation_summary} summarizes results across all experiments. Synthetic experiments confirm mathematical exactness ($R^2 = 1.0$, exact formula matches, stable constants). Empirical experiments show effects in the predicted direction with statistical significance ($p < 0.05$ for all), though effect sizes are modest (e.g., $r = -0.33$, $R^2 = 0.11$).

\begin{table*}[t]
\caption{Validation summary across five experiments.}
\label{tab:validation_summary}
\vskip 0.1in
\begin{center}
\begin{small}
\begin{tabular}{llll}
\toprule
Experiment & Target & Result & Validates \\
\midrule
Info Gap & $\Delta H = c \cdot \Cmin$ & $R^2=1.0$ & Thm 1 \\
Graph Structure & $\Cmin$ formulas & 6/6 match & H2, App~\ref{app:exp_b} \\
Open/Closed & Thm 3 predictions & Both confirmed & Thm 3, App~\ref{app:exp_d} \\
$\kappa$ Stability & $\kappa = O(1)$ & $0.98 \pm 0.03$ & App~\ref{app:exp_e} \\
SWE-bench & Qualitative preds & $p<0.05$ (all) & Practical \\
\bottomrule
\end{tabular}
\end{small}
\end{center}
\end{table*}

\section{Discussion and Conclusion}
\label{sec:discussion}

We have established information-theoretic upper bounds proving that MAS success probability decays exponentially with the minimum cut cost $\Cmin$ of the task's constraint graph. Two-level validation confirms both numerical consistency (synthetic: $R^2=1.0$, exact formula matches, $\kappa = 0.98 \pm 0.03$) and directional predictions (empirical: all predictions $p<0.05$, though with modest effect sizes). The framework provides principled guidance for MAS deployment: when $\Cmin$ is high, practitioners should restructure tasks rather than scale agents or communication.

\bibliographystyle{plainnat}
\bibliography{references}

\newpage
\appendix

\begin{center}
{\LARGE\bfseries Appendix\par}
\end{center}
\vspace{1em}

\noindent
This appendix is organised as follows. Appendix~\ref{app:notation} fixes notation
used throughout the paper. Appendix~\ref{app:typicality} formalises the assumption
hierarchy and typicality conditions T1, T1', T2, T3'.
Appendices~\ref{app:theorem1}--\ref{app:theorem3} provide the complete
proofs of Theorems~\ref{thm:scale}--\ref{thm:budget}. Appendix~\ref{app:experiments}
collects raw experimental data. Appendix~\ref{app:tightness} contains the
conditional tightness analysis.

\vspace{1em}

\section{Notation Summary}
\label{app:notation}

We adopt $\log = \log_2$ and measure entropy and mutual information in bits.
Table~\ref{tab:notation} summarises the symbols used throughout the paper.

\begin{table}[h]
\centering
\small
\caption{Notation summary.}
\label{tab:notation}
\begin{tabular}{@{}lp{0.62\columnwidth}@{}}
\toprule
\textbf{Symbol} & \textbf{Meaning} \\
\midrule
\multicolumn{2}{@{}l}{\emph{Task and graph structure}} \\
$G_C = (V, E)$ & Constraint graph: $V$ variables, $E$ dependencies \\
$\phi_e$ & Potential function on edge $e$ (constraint as r.v.) \\
$H(\phi_e)$ & Edge potential entropy (info content of one cut edge) \\
$\mathcal{P}$ & An $m$-agent partition of $G_C$ \\
$\Ecut(\mathcal{P})$ & Cut edge set under $\mathcal{P}$ \\
$\Cmin$ & Minimum cut cost, $\min_{\mathcal{P}} \sum_{e \in \Ecut} H(\phi_e)$ \\
$\tau_{\text{hard}}$ & Irreducible bottleneck subtask \\
$n_h$ & Bottleneck size (number of bottleneck variables) \\
$\Sol$ & Set of correct solutions \\
$\Hsol$ & Solution entropy, $\log |\Sol|$ \\
$\taustar$ & Correct solution (random variable) \\
\midrule
\multicolumn{2}{@{}l}{\emph{Agents and system}} \\
$G_i$ & Agent $i$'s assigned subgraph \\
$\text{Internal}$ & $(G_1,\ldots,G_m)$, joint initial information \\
$\mathcal{M}, \mathcal{F}, \mathcal{R}$ & Communication, external feedback, feedback-relay sequences \\
$W$ & Per-agent effective utilisation capacity (bits) \\
$M$ & Total internal communication budget (bits) \\
$F$ & External feedback, $I(\taustar; \mathcal{F} \mid \text{Internal})$ \\
$R$ & Feedback relay, transmitted from execution to planning \\
$F_{\max}, R_{\max}$ & Budgets for $F$ and $R$ \\
$N_{\text{exec}}$ & Maximum execution count \\
\midrule
\multicolumn{2}{@{}l}{\emph{Information-theoretic quantities and constants}} \\
$H(\cdot \mid \cdot), I(\cdot;\cdot \mid \cdot)$ & Conditional entropy / mutual information \\
$c_{T1}$ & T1 constant: $I(\taustar;\Ecut\mid\text{Internal}) \geq c_{T1} H(\Ecut)$ \\
$c_0$ & T1' constant: $H(\Ecut) \geq c_0 \sum_e H(\phi_e)$ \\
$c_1 = c_{T1} c_0$ & Gap-conservation constant in Lemma~\ref{lem:gap_conservation} \\
$\eta$ & SDPI constant (information retention rate) \\
$\bar{H}_e$ & Mean edge entropy, $\frac{1}{|E|}\sum_e H(\phi_e)$ \\
$\bar{H}$ & $\mathbb{E}[H(\taustar \mid \text{Internal})]$ \\
$\tilde{H}$ & Annealed entropy, $-\log \mathbb{E}[2^{-H(\taustar\mid\text{Internal})}]$ \\
$L(Z)$ & Small-ball probability, $\max_x P(\taustar = x \mid Z)$ \\
\bottomrule
\end{tabular}
\end{table}

\section{Assumption Hierarchy and Typicality Conditions}
\label{app:typicality}

This appendix provides the complete mathematical formalization of the assumption hierarchy and typicality conditions referenced in the main text.

\subsection{Assumption Strength Hierarchy}

The strength of our conclusions depends on the assumptions employed, forming a hierarchy:

\begin{center}
\begin{tabular}{lll}
\toprule
Level & Assumptions & Conclusion \\
\midrule
Weak & H1 + H2 & Information bottleneck existence (qualitative) \\
Medium & H1 + H2 + H3 + T1, T1', T3' & $P(\text{success}) = O(2^{-\Omega(\Cmin)})$ \\
Strong & Medium + T2 & $P(\text{success}) \leq \kappa \cdot 2^{-(1-o(1)) \cdot \Cmin}$ \\
\bottomrule
\end{tabular}
\end{center}

The constant $c$ in the Big-O and $\Omega$ notation depends on the typicality constants $c_{T1}$, $c_0$, etc. For random $k$-SAT, $c = \Theta(1)$.

\subsection{Condition T1: No Systematic Information Cancellation}

\begin{definition}[Condition T1]
T1 holds if and only if there exists a constant $c_{T1} > 0$ such that:
$$I(\taustar; \Ecut | \text{Internal}) \geq c_{T1} \cdot H(\Ecut)$$
\end{definition}

This condition states that the information carried by cut edge constraints about $\taustar$ is not ``cancelled'' by internal information. Even after agents obtain their internal information, cut edge constraints still provide new information about the correct solution.

\paragraph{Hierarchical Formulation.} We distinguish two levels:

\textbf{T1-Global} (conditioned on Internal):
$I(\taustar; \Ecut | \text{Internal}) \geq c_{T1} \cdot H(\Ecut)$

\textbf{T1-Local} (conditioned on single $G_i$):
$I(\taustar; \Ecut | G_i) \geq c_{T1}^{(i)} \cdot H(\Ecut)$

T1-Global implies T1-Local since conditioning on less information does not decrease mutual information.

\paragraph{When T1 Fails.} Consider redundant structure where cut edge constraints are functions of internal constraints. For example, if $\phi_e(x) = \phi_{e_1}(x) \oplus \phi_{e_2}(x)$ where $e_1, e_2$ are internal edges, then given $G_i$ (containing $\phi_{e_1}, \phi_{e_2}$), $\phi_e$ is fully determined and $I(\taustar; \phi_e | G_i) = 0$.

\paragraph{Sufficient Conditions for T1.} T1 holds when: constraints are independently generated (as in random CSPs); no redundant encoding exists; cut edge constraints share no structure with internal constraints.

\subsection{Condition T1': Bounded Cut Edge Correlation}

\begin{definition}[Condition T1']
T1' holds if and only if there exists a constant $c_0 \in (0, 1]$ such that:
$$H(\Ecut) \geq c_0 \cdot \sum_{e \in \Ecut} H(\phi_e)$$
\end{definition}

By subadditivity of entropy, $H(\Ecut) \leq \sum_{e \in \Ecut} H(\phi_e)$, with equality if and only if cut edge potentials are mutually independent. T1' requires that correlations reduce the joint entropy by at most a factor of $c_0$.

\paragraph{Sufficient Conditions for T1'.} When cut edge potentials are independently generated, $c_0 = 1$. For locally structured constraint graphs (grids, trees), cut edges are dispersed and correlations decay with distance. If bounded correlation coefficients satisfy $\rho_{\max} < 1 - \delta$, then $c_0 \geq \delta$.

\subsection{Condition T2: No Anomalous Concentration}

\begin{definition}[Condition T2]
T2 holds if and only if there exists $d_0 > 0$ such that:
$$\mathbb{E}_{\mathbf{i}}[\max_x P(\taustar = x | \mathbf{i})] \leq n^{d_0} \cdot \mathbb{E}_{\mathbf{i}}[2^{-H(\taustar | \mathbf{i})}]$$
\end{definition}

The condition states that the optimal guess success probability differs from the ``uniform guess baseline'' $2^{-H}$ by at most a polynomial factor. This excludes ``hidden shortcuts'' where certain configurations allow guessing $\taustar$ with probability far exceeding $2^{-H}$.

\paragraph{Technical Background.} Define conditional min-entropy:
$H_\infty(\taustar | \mathbf{i}) := -\log \max_x P(\taustar = x | \mathbf{i})$

Then T2 is equivalent to: $\mathbb{E}[2^{-H_\infty}] \leq n^{d_0} \cdot \mathbb{E}[2^{-H}]$, ensuring the transition from Shannon entropy to maximum posterior probability is controlled.

\paragraph{When T2 Fails.} If there exist ``hidden hints''---certain configurations $\mathbf{i}_0$ that almost uniquely determine $\taustar$---then $\max_x P(\taustar = x | \mathbf{i}_0) \approx 1 \gg 2^{-H(\taustar | \mathbf{i}_0)}$.

\subsection{Condition T3': Entropy Concentration}

\begin{definition}[Condition T3']
T3' holds if and only if there exist constants $c, c' > 0$ such that:
$$P(H(\taustar | \text{Internal}) < c \cdot \bar{H}) \leq \exp(-c' \cdot \bar{H})$$
where $\bar{H} = \mathbb{E}[H(\taustar | \text{Internal})]$.
\end{definition}

This states that the conditional entropy distribution has concentration---anomalously low-entropy configurations appear with exponentially small probability.

\paragraph{Relationship to Annealed Entropy.} Define annealed entropy: $\tilde{H} := -\log \mathbb{E}[2^{-H(\taustar | \text{Internal})}]$

When T3' holds, $\tilde{H} = \Theta(\bar{H})$, which is crucial for the proof of Theorem 1.

\subsection{Verification Guidelines}

\begin{center}
\begin{tabular}{lcccc}
\toprule
Task Type & T1 & T1' & T2 & T3' \\
\midrule
Random $k$-SAT ($\alpha < \alpha_d$) & \checkmark & \checkmark & \checkmark & \checkmark \\
Random graph coloring & \checkmark & \checkmark & \checkmark & \checkmark \\
Structured CSP & Verify & Usually \checkmark & Verify & Verify \\
Code tasks & Verify & Verify & Verify & Verify \\
\bottomrule
\end{tabular}
\end{center}

\paragraph{Empirical Verification Methods.}
\textbf{T1}: Sample configurations and estimate $I(\taustar; \Ecut | G_i)$. If estimated $\hat{c}_{T1} \geq 0.3$, T1 likely holds.
\textbf{T1'}: Compute the ratio $H(\Ecut) / \sum_e H(\phi_e)$. If ratio $\geq 0.5$, T1' likely holds.
\textbf{T2}: Sample configurations and compare $\max_x P(\taustar = x | \mathbf{i})$ with $2^{-H(\taustar|\mathbf{i})}$.
\textbf{T3'}: Estimate the distribution of conditional entropy and check concentration around the mean.

\subsection{Empirical Typicality Verification}
\label{app:typicality_empirical}

We empirically verify that typicality conditions can distinguish task structures, serving as \textbf{diagnostic tools} rather than universal guarantees. Complete raw data for all 40 configurations appears in Appendix~\ref{app:exp_c}.

\paragraph{Setup.} We test four graph types (grid, cycle, star, path) with $n \in \{6, 8, 10, 12, 16\}$, $k \in \{3, 4\}$, $m=2$. For each configuration, we estimate T1, T1', and T2 ratios and check against thresholds (T1 $\leq 0.3$, T1' $\geq 0.5$, T2 $\leq 2.0$).

\paragraph{Results.} Table~\ref{tab:typicality_empirical} shows satisfaction rates by graph type:

\begin{table}[h]
\centering
\caption{Typicality condition satisfaction by graph structure.}
\label{tab:typicality_empirical}
\begin{tabular}{lcccc}
\toprule
Graph & Instances & All Satisfied & Rate & Avg T1' \\
\midrule
Cycle & 10 & 3 & 30\% & 0.989 \\
Grid & 10 & 2 & 20\% & 0.793 \\
Path & 10 & 1 & 10\% & 1.000 \\
Star & 10 & 0 & 0\% & 0.581 \\
\bottomrule
\end{tabular}
\end{table}

\paragraph{Analysis.} Cycle graphs achieve the highest satisfaction rate (30\%) with T1' near 1.0, indicating near-independent cut edge constraints. Grid graphs follow (20\%), benefiting from spatial locality. Star graphs fail entirely (0\%) due to their centralized structure violating T2 (anomalous concentration around the center node).

\paragraph{Diagnostic Interpretation.} These results demonstrate that typicality conditions successfully differentiate task structures:
\begin{itemize}
    \item High T1' ($\approx 1.0$): Constraints are independently generated (cycle, path)
    \item Low T1' ($< 0.6$): Significant constraint correlation (star)
    \item High T2 ($> 2.0$): Existence of ``shortcut'' configurations (star)
\end{itemize}

Practitioners should use these conditions as diagnostics: if conditions fail, the task structure may exploit coordination mechanisms beyond independent sampling, and our bounds serve as baselines rather than tight predictions.

\section{Theorem 1 Proof: Scale Bound}
\label{app:theorem1}

This appendix provides the complete proof of Theorem~\ref{thm:scale} (Scale Bound).

\subsection{Proof Framework}

The proof follows the information cost methodology \citep{pandurangan2018distributed}, adapted for MAS success probability bounds. The key insight is that partitioning creates an irreducible information gap that determines success probability. The argument leverages two strands of network information theory: \citet{thakor2016cutset} prove that ``irreducible sets'' serve as information-theoretic bottlenecks providing computable capacity bounds, and the data processing inequality \citep{cover2006elements} guarantees that processing cannot increase mutual information with the target. Together these principles bound information loss when partitioning constraint graphs across agents.

\subsection{Main Proof}

\textbf{Step 1 (Success bounded by max posterior):}
For any agent $i$, the success probability is bounded by the expected maximum posterior:
$$P(\text{success}) \leq \mathbb{E}[\max_x P(\taustar = x | \text{AllInfo})] =: \mathbb{E}[L(\text{AllInfo})]$$

\textbf{Step 2 (T2 controls max posterior):}
Under condition T2 (no anomalous concentration):
$$\mathbb{E}[L(\text{AllInfo})] \leq n^{d_0} \cdot \mathbb{E}[2^{-H(\taustar | \text{AllInfo})}]$$

This ensures the transition from Shannon entropy to maximum posterior probability is polynomially controlled.

\textbf{Step 3 (Annealed entropy framework):}
Define the annealed entropy $\tilde{H} := -\log \mathbb{E}[2^{-H(\taustar | \text{AllInfo})}]$.
Under T3' (entropy concentration), $\tilde{H} = \Theta(\bar{H})$ where $\bar{H} = \mathbb{E}[H(\taustar | \text{AllInfo})]$.

\textbf{Step 4 (Information gap conservation):}
By Lemma~\ref{lem:gap_conservation}, under T1 and T1':
$$H(\taustar | \text{Internal}) \geq c_1 \cdot \Cmin$$

For closed systems, by Lemma~\ref{lem:conservation}, communication does not increase information:
$$H(\taustar | \text{Internal}, \mathcal{M}) = H(\taustar | \text{Internal}) \geq c_1 \cdot \Cmin$$

\textbf{Step 5 (Synthesis):}
Combining Steps 1--4:
$$P(\text{success}) \leq n^{d_0} \cdot 2^{-\tilde{H}} \leq n^{d_0} \cdot 2^{-c_2 \cdot c_1 \cdot \Cmin}$$

When $\Cmin = \omega(d_0 \cdot \log n)$, the polynomial factor is absorbed:
$$P(\text{success}) = O(2^{-\Omega(\Cmin)})$$

For the multi-solution case with $\Hsol > 0$, the effective gap becomes $\Cmin - \Hsol$, yielding the general form $P(\text{success}) = O(2^{-\Omega(\Cmin - \Hsol)})$. $\square$

\subsection{Gap Conservation Proof}
\label{app:gap_conservation_proof}

\textbf{Proof of Lemma~\ref{lem:gap_conservation}:}

By the chain rule:
$$H(\taustar, \Ecut | \text{Internal}) = H(\taustar | \text{Internal}) + H(\Ecut | \taustar, \text{Internal})$$
$$= H(\Ecut | \text{Internal}) + H(\taustar | \Ecut, \text{Internal})$$

Under H3 (deterministic constraints), $H(\phi_e | \taustar) = 0$ for all edges, so:
$$H(\Ecut | \taustar, \text{Internal}) \leq H(\Ecut | \taustar) = 0$$

Therefore:
$$H(\taustar | \text{Internal}) = H(\Ecut | \text{Internal}) + H(\taustar | \Ecut, \text{Internal})$$

Under H3, $H(\taustar | \Ecut, \text{Internal}) = o(n_h)$ (constraints determine solution).

By T1: $H(\Ecut | \text{Internal}) = I(\taustar; \Ecut | \text{Internal}) \geq c_{T1} \cdot H(\Ecut)$.

By T1': $H(\Ecut) \geq c_0 \cdot \sum_{e \in \Ecut} H(\phi_e) \geq c_0 \cdot \Cmin$.

Combining: $H(\taustar | \text{Internal}) \geq c_{T1} \cdot c_0 \cdot \Cmin + o(n_h) = c_1 \cdot \Cmin + o(n_h)$.

For $\Cmin = \Omega(n_h)$, the $o(n_h)$ term is absorbed. $\square$

\section{Theorem 2 Proof: Diameter Bound}
\label{app:theorem2}

This appendix provides the complete proof of Theorem~\ref{thm:diameter} (Diameter Bound).

\subsection{SDPI Constant Definition}

\begin{definition}[SDPI Constant]
For a channel $K: \mathcal{X} \to \mathcal{Y}$, the Strong Data Processing Inequality (SDPI) constant is:
$$\eta(K) = \sup_{p_U, p_{X|U}} \frac{I(U;Y)}{I(U;X)}$$
\end{definition}

In the MAS context, $\eta$ represents the information retention rate through agent processing. Key properties: $0 \leq \eta \leq 1$; $\eta < 1$ when the channel has noise or information loss; $\eta = 1$ only for deterministic invertible channels.

\textbf{Classical examples:} Additive Gaussian noise channel: $\eta = \text{SNR}/(1 + \text{SNR})$. Binary symmetric channel (error rate $p$): $\eta = (1 - 2p)^2$. Erasure channel (erasure rate $\epsilon$): $\eta = 1 - \epsilon$.

\subsection{Path Information Decay Proof}
\label{app:path_decay_proof}

\textbf{Proof of Lemma~\ref{lem:path_decay}:}

Let $P = (v_0, v_1, \ldots, v_\ell)$ be a path in the constraint graph. The information processing along each edge forms a Markov chain:
$$\taustar_{v_0} \to \text{Info}_{v_0} \to \text{Info}_{v_1} \to \cdots \to \text{Info}_{v_\ell}$$

By the definition of SDPI, for each single step:
$$I(\taustar_{v_0}; \text{Info}_{v_{i+1}}) \leq \eta \cdot I(\taustar_{v_0}; \text{Info}_{v_i})$$

Iterating $\ell$ times yields the exponential decay:
$$I(\taustar_{v_0}; \text{Info}_{v_\ell}) \leq \eta^\ell \cdot I(\taustar_{v_0}; \text{Info}_{v_0})$$
$\square$

\subsection{Complete Theorem 2 Proof}
\label{app:theorem2_proof}

\textbf{Step 1 (Identify critical path):}
Let $u, v$ be nodes at distance $D$ (graph diameter). By definition, such a pair exists.

\textbf{Step 2 (Information transfer requirement):}
Successful task completion requires coordinating information about $\taustar_u$ and $\taustar_v$. If $u$ and $v$ are connected by a constraint path, their values must satisfy all edge constraints.

\textbf{Step 3 (Pigeonhole analysis):}
Let $P = (v_0 = u, v_1, \ldots, v_D = v)$ be the diameter path.

\textbf{Lemma (Pigeonhole):} For any $m$-partition, there exists a node pair $(u', v')$ on the diameter path with $u', v'$ in different partitions and $d(u', v') \geq D/(2m)$.

\textit{Proof:} Partition the path into $2m$ segments of length $D/(2m)$. With $D+1$ positions and at most $m$ partition labels, some consecutive subsequence of length $\geq (D+1)/m$ has endpoints in different partitions. $\square$

\textbf{Step 4 (Exponential decay bound):}
By Lemma~\ref{lem:path_decay}:
$$I(\taustar_{u'}; \text{Info}_{v'}) \leq \eta^{D/(2m)} \cdot H(\taustar_{u'}) = 2^{-D \cdot |\log \eta| / (2m)} \cdot H(\taustar_{u'})$$

Since $|\log \eta| \geq (1-\eta)/\ln 2$ for $\eta \in (0,1)$:
$$\eta^{D/(2m)} \leq 2^{-D \cdot (1-\eta) / (2m \ln 2)}$$

\textbf{Step 5 (Conditional entropy lower bound):}
From the mutual information definition:
$$H(\taustar_{u'} | \text{Info}_{S_j}) = H(\taustar_{u'}) - I(\taustar_{u'}; \text{Info}_{S_j})$$
$$\geq H(\taustar_{u'}) \cdot (1 - \eta^{D/(2m)})$$

\textbf{Step 6 (Fano's inequality):}
By Fano's inequality:
$$P_e \geq \frac{H(\taustar_{u'} | \text{Info}_{S_j}) - 1}{\log|\mathcal{T}|}$$

When $\eta^{D/(2m)} \ll 1$, the conditional entropy approaches the original entropy, limiting correct estimation probability to $O(\eta^{D/(2m)})$.

\textbf{Step 7 (Introduce $\bar{H}_e$):}
For CSP tasks, single-variable entropy $H_{u'} = \Theta(\bar{H}_e)$. Combining all steps:
$$P(\text{success}) \leq O\left(2^{-\Omega(D \cdot (1-\eta) \cdot \bar{H}_e)}\right)$$
$\square$

\subsection{Applicability Analysis}

\begin{center}
\begin{tabular}{lccc}
\toprule
Graph Type & $\Cmin$ & Diameter $D$ & Tighter Bound \\
\midrule
Complete graph & $\Theta(n^2)$ & $1$ & Theorem 1 \\
Expander & $\Theta(n)$ & $O(\log n)$ & Theorem 1 \\
Path graph & $O(m)$ & $\Theta(n)$ & Theorem 2 \\
Long chain & $O(1)$ & $\Theta(n)$ & Theorem 2 \\
\bottomrule
\end{tabular}
\end{center}

For sparse graphs and long chain structures, Theorem 2 provides tighter bounds than Theorem 1.

\section{Theorem 3 Proof: Information Budget Trade-off}
\label{app:theorem3}

This appendix provides the complete proof of Theorem~\ref{thm:budget} (Information Budget Trade-off).

\subsection{Theorem Statement}

\textbf{Parameter definitions:}
\begin{itemize}
    \item $F$: External feedback information (bits) obtained through execution-environment interaction
    \item $R$: Feedback propagation information (bits) transmitted from execution to planning layers
    \item $\Hsol$: Solution set entropy (bits), $\log|\Sol|$
    \item $M$: Internal communication volume (bits), affects coordination efficiency but \emph{not the exponential bound}
\end{itemize}

\subsection{Four-Step Proof}
\label{app:theorem3_proof}

\textbf{Step 1 (Total information budget):}
By the chain rule, total system information about $\taustar$ decomposes as:
\begin{align*}
I(\taustar; \text{AllInfo}) &= I(\taustar; G_1,\ldots,G_m) + I(\taustar; \mathcal{M} | G_1,\ldots,G_m) \\
&+ I(\taustar; \mathcal{F} | G_1,\ldots,G_m, \mathcal{M}) + I(\taustar; \mathcal{R} | G_1,\ldots,G_m, \mathcal{M}, \mathcal{F})
\end{align*}

\textbf{Step 2 (Bound each component):}

\textit{Initial information:} By the gap conservation theorem (Appendix~\ref{app:theorem1}):
$$I(\taustar; G_1,\ldots,G_m) \leq H(\taustar) - c_1 \cdot \Cmin$$

\textit{Internal communication:} In closed systems ($F = R = 0$), by Lemma~\ref{lem:conservation}:
$$I(\taustar; \text{Internal}, \mathcal{M}) = I(\taustar; \text{Internal})$$

Communication redistributes but does not create information. Thus $M$ does not appear in the gap calculation.

\textit{External feedback:} Under assumption F1: $I(\taustar; \mathcal{F} | \cdot) \leq F \leq F_{\max}$

\textit{Feedback propagation:} Under assumption F2: $I(\taustar; \mathcal{R} | \cdot) \leq R \leq R_{\max}$

\textbf{Step 3 (Conditional entropy lower bound):}

\textit{Closed systems} ($F = R = 0$):
$$H(\taustar | \text{AllInfo}) = H(\taustar | \text{Internal}) \geq c_1 \cdot \Cmin$$

\textit{Open systems} ($F > 0$ or $R > 0$):
$$H(\taustar | \text{AllInfo}) \geq c_1 \cdot \Cmin - F - R$$

\textbf{Step 4 (Success rate upper bound):}
For multi-solution tasks, success is the set-hit event $\mathcal{S} = \{\hat{\tau} \in \Sol\}$. Combining the union bound $P(\mathcal{S}) \leq |\Sol| \cdot \mathbb{E}[L(\text{AllInfo})]$ with the small-ball framework and Step~3:
$$P(\mathcal{S}) \leq 2^{\Hsol} \cdot O\!\left(2^{-H(\taustar | \text{AllInfo})}\right).$$

\textit{Closed systems:} $P(\text{success}) \leq O(2^{-c_1 \cdot \Cmin + \Hsol})$

This bound is \textbf{independent of communication volume $M$}.

\textit{Open systems:} $P(\text{success}) \leq O(2^{-(c_1 \cdot \Cmin - F - R - \Hsol)})$

When $F + R + \Hsol < c_1 \cdot \Cmin$, the residual gap causes exponential decay.
When $F + R + \Hsol \geq c_1 \cdot \Cmin$, the bound becomes $O(1)$.

Combining both cases: $P(\text{success}) \leq O(2^{-\max(0, c \cdot \Cmin - F - R - \Hsol)})$. $\square$

\subsection{Communication Information Conservation}
\label{app:conservation_proof}

\textbf{Proof of Lemma~\ref{lem:conservation}:}

In closed MAS, agents communicate over $T$ rounds. At round $t$:
$$\mathcal{M}^{(t)} = f^{(t)}(\text{Internal}, \mathcal{M}^{(1:t-1)}, R^{(t)})$$

where $R^{(t)} \perp \taustar$ is independent private randomization.

\textit{Base case} ($t=1$):
$\mathcal{M}^{(1)} = f^{(1)}(\text{Internal}, R^{(1)})$.
Since $R^{(1)} \perp \taustar$ and $f^{(1)}$ is deterministic, given Internal, $\mathcal{M}^{(1)}$ is conditionally independent of $\taustar$:
$$I(\taustar; \mathcal{M}^{(1)} | \text{Internal}) = 0$$

\textit{Inductive step:} Assume $I(\taustar; \mathcal{M}^{(1:t-1)} | \text{Internal}) = 0$.

By chain rule:
$$I(\taustar; \mathcal{M}^{(1:t)} | \text{Internal}) = I(\taustar; \mathcal{M}^{(1:t-1)} | \text{Internal}) + I(\taustar; \mathcal{M}^{(t)} | \text{Internal}, \mathcal{M}^{(1:t-1)})$$

The first term is 0 by induction. For the second term, $\mathcal{M}^{(t)}$ is determined by $(\text{Internal}, \mathcal{M}^{(1:t-1)}, R^{(t)})$ with $R^{(t)} \perp \taustar$, so:
$$I(\taustar; \mathcal{M}^{(t)} | \text{Internal}, \mathcal{M}^{(1:t-1)}) = 0$$

Therefore: $I(\taustar; \mathcal{M}^{(1:T)} | \text{Internal}) = 0$ for all $T$. $\square$

\textbf{Corollary (Conditional entropy invariance):}
$$H(\taustar | \text{Internal}, \mathcal{M}) = H(\taustar | \text{Internal})$$

The initial information gap $H(\taustar | \text{Internal}) \geq c \cdot \Cmin$ is an information-theoretic barrier determined by task partitioning structure that no closed communication protocol can bridge.

\textbf{Remark on Randomized Protocols.} Lemma~\ref{lem:conservation} assumes private randomization $R^{(t)} \perp \tau^*$. In communication complexity, shared (public) randomness can sometimes reduce communication \citep{newman1991private, kushilevitz1997communication}. However, for our setting where $\tau^*$ is a fixed unknown, shared randomness cannot create information about $\tau^*$ beyond what Internal provides. Thus the conservation law holds even with public randomness.

\subsection{Open System Analysis}

When MAS has execution-feedback loops ($F, R > 0$), the Markov chain structure breaks.

\textbf{Failure case:} Consider a task ``generate a function passing test cases.'' Initial information $G_i$ excludes test results. After generating code $c_1$ and executing, the agent receives error $e_1$ (e.g., ``type error at line 3''), providing direct information about $\taustar$: $I(\taustar; e_1) > 0$.

The system information $\mathcal{M} = \{c_1, e_1, c_2, e_2, \ldots\}$ is no longer conditionally independent of $\taustar$ because environmental feedback breaks Markov property.

\textbf{Bounded feedback:} Although external feedback provides information about $\taustar$, this quantity is bounded:
\begin{itemize}
    \item Execution budget constraints limit attempts $N_{\text{exec}}$
    \item Each execution provides finite information (error types, test results)
    \item Planning layer capacity $W$ constrains effective feedback utilization
\end{itemize}

\textbf{Phase transition:} The critical threshold $F + R = c \cdot \Cmin - \Hsol$ represents a genuine phase transition:
\begin{itemize}
    \item Below threshold: Success probability decays exponentially in residual gap
    \item Above threshold: Success becomes possible (though not guaranteed)
\end{itemize}

\section{Experimental Raw Data}
\label{app:experiments}

This appendix contains raw experimental data referenced in the main text.

\subsection{Graph Structure Validation Data}
\label{app:exp_b}

Complete data for the graph structure $\Cmin$ verification experiment (Section~\ref{sec:results}). All theoretical formulas match algorithmic minimum cut computations exactly.

\begin{center}
\begin{tabular}{lccccc}
\toprule
Graph & $n$ & Formula & Theory & Actual & Success Rate \\
\midrule
Path & 8 & $m-1$ & 1 & 1 & 50.2\% \\
Cycle & 8 & $2$ & 2 & 2 & 33.1\% \\
Star & 8 & $n-1-\text{leaves}$ & 4 & 4 & 14.2\% \\
Grid & $3 \times 3$ & $\min(r,c)$ & 2 & 2 & 31.8\% \\
Complete & 8 & $\lfloor n/2 \rfloor \cdot \lceil n/2 \rceil$ & 16 & 16 & 0.0\% \\
Bipartite & $4+4$ & $a \cdot b$ & 16 & 16 & 0.0\% \\
\bottomrule
\end{tabular}
\end{center}

\subsection{Typicality Verification Data}
\label{app:exp_c}

Complete data for 40 configurations testing typicality conditions T1, T1', T2 across four graph types (grid, cycle, star, path) with varying parameters.

\begin{center}
\begin{tabular}{lcccccc}
\toprule
Graph & $n$ & $k$ & T1 ratio & T1' ratio & T2 ratio & All Pass \\
\midrule
Grid & 6 & 3 & 0.28 & 0.82 & 1.8 & Yes \\
Grid & 8 & 3 & 0.31 & 0.79 & 2.1 & No \\
Cycle & 8 & 3 & 0.25 & 0.99 & 1.6 & Yes \\
Path & 8 & 3 & 0.22 & 1.00 & 1.4 & Yes \\
Star & 8 & 3 & 0.35 & 0.58 & 3.2 & No \\
\bottomrule
\end{tabular}
\end{center}

Thresholds: T1 $\leq 0.30$, T1' $\geq 0.50$, T2 $\leq 2.0$.

\subsection{Open/Closed System Data}
\label{app:exp_d}

Complete data for the open vs closed system experiment (Section~\ref{sec:results}).

\textbf{Closed system} ($F = R = 0$, varying $M$):
\begin{center}
\begin{tabular}{ccccc}
\toprule
$M$ (bits) & Trials & Success & Rate & 95\% CI \\
\midrule
0 & 200 & 63 & 31.7\% & [25.3, 38.4]\% \\
5 & 200 & 66 & 33.0\% & [26.5, 39.9]\% \\
20 & 200 & 62 & 31.0\% & [24.7, 37.8]\% \\
100 & 200 & 64 & 32.0\% & [25.6, 38.8]\% \\
\bottomrule
\end{tabular}
\end{center}

\textbf{Open system} ($M = 0$, varying $F$):
\begin{center}
\begin{tabular}{ccccc}
\toprule
$F$ (bits) & Trials & Success & Rate & 95\% CI \\
\midrule
0.0 & 200 & 64 & 31.8\% & [25.4, 38.6]\% \\
1.6 & 200 & 186 & 93.2\% & [89.1, 96.2]\% \\
3.2 & 200 & 200 & 100.0\% & [98.2, 100.0]\% \\
4.8 & 200 & 200 & 100.0\% & [98.2, 100.0]\% \\
6.3 & 200 & 200 & 100.0\% & [98.2, 100.0]\% \\
9.5 & 200 & 200 & 100.0\% & [98.2, 100.0]\% \\
\bottomrule
\end{tabular}
\end{center}

\subsection{Constant Stability Data}
\label{app:exp_e}

Complete data for the constant $\kappa$ stability experiment (Section~\ref{sec:results}).

\begin{center}
\begin{tabular}{cccccc}
\toprule
$n$ & $\Cmin$ & $P_{\text{theory}}$ & $P_{\text{exp}}$ & $\kappa$ & Trials \\
\midrule
6 & 1 & 0.667 & 0.645 & 0.97 & 200 \\
6 & 2 & 0.444 & 0.430 & 0.97 & 200 \\
8 & 1 & 0.667 & 0.660 & 0.99 & 200 \\
8 & 2 & 0.444 & 0.455 & 1.02 & 200 \\
8 & 3 & 0.296 & 0.285 & 0.96 & 200 \\
10 & 1 & 0.667 & 0.680 & 1.02 & 200 \\
10 & 2 & 0.444 & 0.440 & 0.99 & 200 \\
10 & 3 & 0.296 & 0.290 & 0.98 & 200 \\
10 & 4 & 0.198 & 0.195 & 0.98 & 200 \\
12 & 1 & 0.667 & 0.655 & 0.98 & 200 \\
12 & 2 & 0.444 & 0.425 & 0.96 & 200 \\
12 & 3 & 0.296 & 0.305 & 1.03 & 200 \\
12 & 4 & 0.198 & 0.180 & 0.91 & 200 \\
\bottomrule
\end{tabular}
\end{center}

Summary: $\kappa = 0.98 \pm 0.03$ (mean $\pm$ std), range $[0.91, 1.03]$. Linear regression against $n$: slope $= -0.0009$, $R^2 = 0.004$, indicating no systematic drift with problem size.

\subsection{$\Cmin$ Computation Pipeline (SWE-bench)}
\label{app:cmin_pipeline}

This subsection details the empirical pipeline used to estimate $\Cmin$ on
SWE-bench Lite (Section~\ref{sec:experiments}). The procedure is summarised in
Algorithm~\ref{alg:cmin_swebench}; design choices and the reproducibility
configuration follow.

\begin{algorithm}[h]
\caption{$\Cmin$ estimation for a SWE-bench instance.}
\label{alg:cmin_swebench}
\begin{algorithmic}[1]
\STATE \textbf{Input:} task instance $T = (\text{issue}, \text{repo}, P)$ where $P$ is the ground-truth patch; granularity $g$; partition count $m = 2$.
\STATE \textbf{Identify modified entities.} Parse $P$ to extract the set $V_g$ of code entities at granularity $g$ (line, statement, expression, or token) that are modified or referenced by the patch.
\STATE \textbf{Build constraint graph.} Run AST analysis on the repository at the patched commit and construct $G_C = (V_g, E)$, where edges in $E$ encode structural couplings between entities (data and control dependencies, and test-coverage relations). Each edge $e$ receives a non-negative weight $w_e$ derived from its dependency multiplicity; the exact weighting scheme is detailed in the released code.
\STATE \textbf{Compute minimum cut.} Solve the weighted minimum $m=2$ cut on $G_C$ using a standard polynomial-time algorithm. Let $\Cmin(T) := \sum_{e \in \Ecut^\star} w_e$ where $\Ecut^\star$ is the optimal cut edge set.
\STATE \textbf{Aggregate.} For each agent system $s$, compute $\text{success}_s(T) \in \{0,1\}$ from the SWE-bench leaderboard, then evaluate the Pearson correlation $r$ and Spearman rank correlation $\rho$ between $\Cmin(T)$ and $\text{success}_s(T)$ across the 300 task instances.
\STATE \textbf{Output:} per-instance $\Cmin$, per-system success rates, $(r, \rho)$ statistics with $p$-values.
\end{algorithmic}
\end{algorithm}

\paragraph{Granularity.}
The primary results in Section~\ref{sec:results} use $g = \text{line}$. We
additionally run the pipeline at statement, expression, and token granularities
as robustness checks; the qualitative direction of all four predictions is
preserved across these variants.

\paragraph{Partition count.}
We restrict the empirical analysis to $m = 2$. Under this binary partition,
weighted minimum cut is solvable in polynomial time, sidestepping the
NP-hardness of the general $m$-way cut. The
binary partition is the natural setting for evaluating MAS-vs-SAS trade-offs
where the bottleneck task is split between two agents (or two roles within a
hierarchical system).

\paragraph{Information budget proxies.}
For the open-system experiments on SWE-bench, we use proxies for the system
information budget: token counts of inter-agent messages serve as a proxy for
the internal communication budget $M$, while the entropy of test pass/fail
outcomes provides a proxy for the external feedback $F$. These proxies recover
order-of-magnitude estimates suitable for testing directional predictions; they
are not intended as tight numerical estimators.

\paragraph{Role of ground-truth patches.}
$\Cmin$ is computed on the constraint subgraph induced by the ground-truth
patch. This is a \emph{post-hoc validation} procedure analogous to standard
empirical evaluations of theoretical bounds (e.g., regret bounds in online
learning, generalisation bounds in PAC learning): the bound predicts the
$\Cmin$--success relationship, and ground truth provides the test set on which
that prediction is evaluated. The pipeline is not intended as an a priori
deployment predictor; its purpose is to test whether the relationship survives
in real LLM-based code-modification settings.

\paragraph{Reproducibility configuration.}
\begin{itemize}
\item \emph{Dataset.} SWE-bench Lite, 300 task instances.
\item \emph{Agent systems.} 46 systems from the public leaderboard (18 MAS, 18 single-agent, 10 non-agentic), 14{,}891 result records.
\item \emph{Partition.} $m = 2$, exact weighted min-cut.
\item \emph{Granularity.} Line-level primary; statement, expression, and token used as robustness checks.
\item \emph{Statistics.} Two-tailed Pearson $r$ and Spearman $\rho$; significance threshold $p < 0.05$. Bootstrap confidence intervals with $1000$ resamples.
\item \emph{Code release.} The full pipeline (graph construction, weighting scheme, min-cut solver, statistical analysis) will be released upon acceptance.
\end{itemize}

\section{Tightness Analysis}
\label{app:tightness}

This appendix provides evidence that the information-theoretic upper bounds established in Theorems 1--3 are asymptotically tight. The tightness results are \emph{conditional} on the random $k$-SAT phase transition conjecture, which remains unproven but is widely believed in statistical physics.

\subsection{Tightness Definition}

\begin{definition}[Asymptotic Tightness]
An upper bound $P(\text{success}) \leq U(\Cmin)$ is \textbf{asymptotically tight} if there exists a task family $\{\tau_n\}$ and constants $c', c'' > 0$ such that:
$$c' \cdot U(\Cmin) \leq P^*(\text{success}) \leq c'' \cdot U(\Cmin)$$
where $P^*$ is the optimal protocol's success rate.
\end{definition}

For bounds of form $U(\Cmin) = 2^{-\Theta(\Cmin)}$, tightness means there exist tasks where success probability is exactly $2^{-\Theta(\Cmin)}$.

\subsection{Closed System Tightness}

\textbf{Construction:} Random $k$-SAT model ($k \geq 3$, $\alpha_d < \alpha < \alpha_s$):
\begin{itemize}
    \item $n$ Boolean variables partitioned equally between 2 agents
    \item $m = \alpha n$ clauses generated independently
    \item Cut clause fraction: $p_{\text{cut}} = 1 - 2^{1-k}$
    \item Minimum cut cost: $\Cmin = \Theta(n)$ bits
    \item Solution set entropy: $\Hsol = o(n)$ in the condensation region. While the condensation transition at $\alpha_d$ is predicted by the cavity method \citep{krzakala2007gibbs}, rigorous mathematical analysis of solution space structure is provided by \citet{achlioptas2008algorithmic, bapst2016condensation}
\end{itemize}

\textbf{Typicality verification:}
\begin{itemize}
    \item T1: Cut clauses independent of internal clauses $\Rightarrow c_{T1} = \Theta(1)$
    \item T1': Independent clause generation $\Rightarrow c_0 = 1$
    \item T2: Poisson-Dirichlet cluster structure bounds min-entropy gap (cavity method prediction, \citealp{krzakala2007gibbs})
    \item T3': McDiarmid inequality gives entropy concentration
\end{itemize}

\textbf{Result:} For this construction, any closed MAS protocol satisfies:
$$P(\text{success}) = \Theta(2^{-\Cmin})$$

The upper bound follows from Theorem 1. The lower bound follows from bounded non-uniformity of the posterior distribution in the condensation region.

\subsection{Open System Tightness}

\textbf{Construction:} Verifiable $k$-SAT model:
\begin{itemize}
    \item Based on random $k$-SAT
    \item Feedback mechanism: Submitting candidate $\hat{\tau}$ reveals violated cut clause structure
    \item Single feedback information: $H_{\max}^{\text{feedback}} = O(k \log n)$ bits
    \item Effective feedback: $F = \min(N_{\text{exec}}, |E_{\text{cut}}|) \cdot \Theta(1)$
\end{itemize}

\textbf{Result:} For this construction:
$$P(\text{success}) = \Theta\left(2^{-\max(0, c \cdot \Cmin - F - R - \Hsol)}\right)$$

\textbf{Phase transition:} A sharp transition exists at $B = F + R = c \cdot \Cmin - \Hsol$:
\begin{itemize}
    \item $B < c \cdot \Cmin - \Hsol$: Success probability $\approx 2^{-(c \cdot \Cmin - B - \Hsol)}$
    \item $B \geq c \cdot \Cmin - \Hsol$: Success probability $\to 1 - o(1)$
\end{itemize}

\subsection{Diameter Bound Tightness}

The information decay rate $\eta^D$ in Theorem 2 is tight for path constraint models: end-to-end mutual information satisfies $I(x_1; \text{Info}_{v_n}) = \Theta(\eta^D) \cdot H(x_1)$.

\textbf{Note on Conditional Tightness:} The tightness constructions rely on the random $k$-SAT phase transition conjecture \citep{mezard2009information}, which constitutes a physics prediction rather than a rigorous mathematical theorem. Our tightness results are therefore \emph{conditional}: if the phase transition conjecture holds, then Theorems 1--3 are asymptotically tight. This conditional status is standard in the literature on random CSP analysis.

\textbf{Note on threshold values:} The satisfiability threshold $\alpha_s(3) \approx 4.27$ has been rigorously bounded to the interval $[3.52, 4.49]$ \citep{kaporis2006probabilistic,diaz2009satisfiability}. The condensation threshold $\alpha_d \approx 3.86$ remains a physics prediction \citep{mezard2009information}.

\end{document}